\documentclass[10pt,b5paper,twoside, headrule]{amsart}

\usepackage{amsfonts, amsmath, amssymb,latexsym}
\usepackage{epsfig}
\usepackage[curve]{xy}
\usepackage{algorithmic}
\usepackage{algorithm}
\usepackage{enumerate}
\usepackage{framed}
\usepackage{hyperref}
\usepackage{mathtools}
\usepackage{lmodern}
\usepackage[T1]{fontenc}

\usepackage{chngcntr}

\footskip=10mm\parskip 0.2cm\addtocounter{page}{0}
\newtheorem{teo}{Theorem}[section]
\newtheorem{lm}[teo]{Lemma}
\newtheorem{prop}[teo]{Proposition}
\newtheorem{cor}[teo]{Corollary}

\newtheorem{ex}[teo]{Example}

\newtheorem{rk}[teo]{Remark}
\theoremstyle{definition}
\newtheorem{df}[teo]{Definition}
\newtheorem*{df*}{Definition}
\newtheorem{nt}[teo]{Notation}

\newtheorem{carg*}{Counterargument}

\theoremstyle{remark}

\newtheorem*{quo*}{Question}

\newcommand{\GR}{\operatorname{GR}}
\newcommand{\cC}{\mathcal{C}}
\newcommand{\CS}{\operatorname{CS}}
\newcommand{\RS}{\operatorname{RS}}
\newcommand{\rank}{\operatorname{rank}}

\newcommand{\binomii}[2]{\genfrac{[}{]}{0pt}{0}{#1}{#2}}
\newcommand{\stirling}[2]{\genfrac{\{}{\}}{0pt}{}{#1}{#2}}
\newcommand{\Mat}{\operatorname{Mat}}
\newcommand{\tr}{\operatorname{Tr}}
\newcommand{\lker}{\operatorname{lker}}
\newcommand{\rker}{\operatorname{rker}}

\def\abs#1{\left\lvert #1 \right\rvert}

\newcommand{\Hom}{\operatorname{Hom}}

\title[MacWilliams duality for rank metric codes]{MacWilliams duality for rank metric codes over finite chain rings}

\begin{document}

\author[I.\,Blanco-Chac\'on]{Iv\'an Blanco-Chac\'on}
\address{Department of Mathematics, Schools of Science, Universidad de Alcal\'a de Henares, Ctra. Madrid-Barcelona Km. 33,600, Alcal\'a de Henares, Spain.}
\email{ivan.blancoc@uah.es}

\author[A.\,F.\,Boix]{Alberto F.\,Boix}
\address{IMUVA--Mathematics Research Institute, Universidad de Valladolid, Paseo de Bel\'en, s/n, 47011, Valladolid, Spain.}
\email{alberto.fernandez.boix@uva.es}

\author[M.\,Greferath]{Marcus Greferath}
\address{School of Mathematics and Statistics, University College Dublin,  D04 V1W8, Dublin 4, Ireland}
\email{marcus.greferath@ucd.ie}

\author[E.\,Hieta-aho]{Erik Hieta--Aho}
\address{VTT Technical Research Centre of Finland Ltd., Visiokatu 4, 33101 Tampere, Finland.\\
Department of Mathematics and Systems Analysis, Aalto University, Espoo, Finland.}
\email{erik.hieta-aho@vtt.fi}


\begin{abstract}
We extend Ravagnani's MacWilliams duality theory to the settings of rank metric codes over finite chain rings, relating the sequences of $q$-binomial moments of a rank metric code over this class of rings with those of its dual.
\end{abstract}



\maketitle

\section{Introduction}
MacWilliams identities relate the weight distribution of a code with that of the corresponding dual code under some inner product. MacWilliams original identity was established in \cite{mcwilliams}, applies to linear error correcting codes over finite fields with the usual Hamming metric, and was soon generalised to non-linear codes over finite fields in \cite{mcwilliams2} under the complete and the Lee metrics.

Further on, Delsarte introduced rank-metric codes in \cite{delsarte}. These codes are linear subspaces of matrices over finite fields, where the distance between two matrices is defined as the rank of their difference. Delsarte's approach interpretes rank-metric codes as association schemes, for which the MacWilliams transform of the weight enumerator is defined in terms of the adjacency algebra of the scheme.

Next, in \cite{gabidulin} Gabidulin proposed a slightly different definition of rank-metric code, in which the codewords are vectors over a finite extension of a finite field (where each coordinate can be replaced by the column vector of its coefficients over the base field obtaining hence a matrix and thus leading to a code in Delsarte's sense). 

However, it was not until 2007 when MacWilliams identities were proved for Gabidulin-like rank-metric codes by Gadouleau and Yan  in \cite{GadouleauYanwithproofs} and in full generality, namely, for Delsarte codes, by Ravagnani in \cite{ravagnani}. Gadouleau and Yan's main contribution is a MacWilliams identity in the form of a closed MacWilliams transform for the original code where a strong use of the $q$-product and $q$-derivative is made, whilst Ravagnani's result exploits the perfect pairing character of the trace bilinear form and several combinatorial properties of the strict shortening operator to obtain a sequence of identities which relates the $q$-binomial moments of the primal and dual codes. 

The interest in linear error correcting codes over finite fields with the Hamming metric is well established, as well as the pursue for codes with the best possible trade-off between rates, minimal distance and asymptotic bounds such as Gilbert-Varshamov's. Regarding rank-metric codes, their most promising application is in the setting of network coding for distributed storage. For instance, they can be used for error detection and correction in cohererent and non-coherent linear network coding under different sets of conditions (\cite{silva1}), as well as in the setting of secure network coding against an eavesdropper (\cite{silva2}).

Possibly, one of the best justifications to study codes over more general alphabets, apart from its intrinsic theoretical interest, is to try to determine, or at least to understand, the maximum size of a code of fixed length over an alphabet of fixed size, such that the minimum pairwise distance between different codewords is lower bounded by a fixed bound. Less ambitious is the seek for more competitive code rates within a given error-correction capability or with just as many codewords \emph{as possible}, or to exploit the natural properties of the Gray mapping between certain types of physical-layer constellations and some binary finite rings. Such non-standard alphabets include number fields and cyclic division algebras \cite{oggierviterbo,hollanti,sethuraman}, orbits of arithmetic Fuchsian groups \cite{bhar1,bhar2} or finite rings, the topic of the present work.

Error-correcting codes over finite rings were first introduced by Preparata \cite{preparata} and Kerdock \cite{kerdock}, and a MacWilliams identity was proved for this family of dual codes. Both codes are non-linear over $\mathbb{F}_2$ but linear over $\mathbb{Z}_4$ with the Lee metric and have been vastly generalised to other families of rings, most prominently for finite chain rings or more in general, to finite Frobenius rings (\cite{HonoldLandjevMacwilliamsidentities}). The reader is referred to \cite{thesis} for a complete study of how the MacWilliams identity generalizes to codes over Frobenius rings under the Hamming metric. Unfortunately, as it is also discussed in loc. cit., only very few families beyond Kerdock-Preparata codes satisfy the identity under the Lee metric. 

It is natural, at least from a theoretical point of view, to investigate similar results for rank-metric codes over finite rings, where one of the first non-trivial problems is to give a satisfactory definition, which comprises, moreover, to choose a suitable definition of the rank. If one is restricted to the class of free modules over finite chain rings, it is still possible to obtain a coherent theory where many features of the rank-metric codes over finite fields still hold (see for instance \cite{rankmetriccodesoverpir}). Regarding \emph{real-world} potential applications of these codes, recent research on nested-lattice-based network coding allows to construct more efficient network coding schemes using rank-metric codes over finite principal ideal rings \cite{gorla,kiran}, for which, up to date it has not been established a MacWilliams identity for the rank-metric weight enumerator. This is the contribution of our present work, which is organized as follows:

In Section \ref{section: preliminaries} we provide an account of definitions and facts on the notions and facts of commutative algebra that we will use along our work. In particular we recall the definitions of chain rings and Galois rings as well as of codes over these, which for us will mean free submodules of a free module of finite rank. We set the notion of rank we will work with, recall the minimal rank distance and the Singleton inequality for finite chain rings.

In Section \ref{section: coefficient MacWilliams} we generalise Ravagnani's duality theory for rank metric codes over finite fields (\cite{ravagnani}) to the setting of rank metric codes over finite chain rings. In particular, thanks to Theorem \ref{the inequality of cardinals holds for chain rings}, we can handle the trace form as in \cite{ravagnani} to prove Lemmas \ref{row space can be viewed as a shortening} and \ref{Lemma 28 ravagnani: correct proof}, with which we obtain Theorem \ref{our MacWilliams identity}, the main result of the section, which relates the binomial moments of the rank distributions of a code over a finite chain ring and its dual.

Along this work, all rings are supposed to be commutative except, of course, the rings of matrices.

\textbf{Acknowledgement.} This work was done in part while I. Blanco-Chac\'on and M. Greferath were visiting professors at the ANTA group at the Department of Mathematics and Systems Analysis, Aalto University, Finland and in part while E. Hieta-aho was a postdoctoral researcher in the same department. The department and the Aalto Science Institute are gratefully acknowledged for their support. I. Blanco-Chac\'on and M. Greferath are partially supported by the Research Council of Finland (project \#351271, PI Camilla Hollanti). I. Blanco-Chac\'on is also partially supported by Spanish Ministerio de Ciencia e Innovaci\'on grant PID2022-136944NB-I00. E. Hieta-aho was partially supported by Research Council of Finland, grant \#336005 (PI Camilla Hollanti). A. F. Boix is partially supported by Spanish Ministerio de Ciencia e Innovaci\'on grant PID2022-137283NB-C22

\section{Preliminaries on modules over finite chain rings}\label{section: preliminaries}

In this section we collect and discuss the basic facts we will use through the paper. We start by introducing the notion of finite chain ring for the convenience of the reader.

\begin{df}[Finite chain rings] A finite chain ring of parameters $(q,s)$ is a finite local ring $R$ with maximal ideal $\mathfrak{m}=R\theta$ for some $\theta\in R$ such that $R/\mathfrak{m}=\mathbb{F}_q$ and such that $R\supset R\theta \supset  \dots \supset R\theta^{s-1} \supset R\theta^{s}={0}$.
\end{df}

An important particular class of chain rings are Galois rings. The interested reader may like to consult \cite[Chapter 3]{Galoisringsbook} for the basics about these rings.

\begin{df}[Galois Ring]\label{Galois ring: definition} A Galois Ring is a ring isomorphic to $\GR(p^s,k):=\mathbb{Z}[x]/(p^s,f(x))$ where $f\in\mathbb{Z}[x]$ is a basic polynomial of degree $k$. By \textit{basic} we mean monic and irreducible over $\mathbb{F}_p[x]$.
\end{df}

In particular, a Galois ring $R$ is a finite chain ring, its unique maximal ideal is $\mathfrak{m}=pR$ and its residue field $\mathbb{F}_{p^k}$.  From now on we will set $q=p^k$ (the cardinality of the residual field) and $\overline{q}=p^s$ (the characteristic of the ring). Likewise we denote $R_1:=\GR(\overline{q},1)$.

By Wilson structure theorems \cite[Theorems A and B]{Wilsonfiniterings}, any finite local ring of characteristic $\overline{q}$ and finite residue field of $q=p^k$ elements contain $\GR (\overline{q},k)$ as a subring. In other words, any  finite local ring may be regarded as a module over a Galois ring. In the case of finite chain rings, more can be said; indeed, it is known that any finite chain ring can be expressed as homomorphic image of a polynomial ring in one variable with coefficients on a Galois ring, see \cite[Theorem XVII.5]{McDonaldFiniteRingsBook}.

\begin{df}[Code over a finite chain ring] Let $R$ be a finite chain ring. A code over $R$ is a free submodule of $M_{m,n}(R)$ for some $m,n\geq 1$.
\end{df}

In what follows, we will denote by $R$ a finite chain ring and given a matrix $A\in M_{m,n}(R)$ we will denote by $\CS (A)$ (respectively, $\RS (A)$) the $R$--submodule of $R^m$ (respectively, $R^n$) generated by the columns (respectively, the rows) of $A.$ Finally, given $M$ a finitely generated $R$--module, we will denote by $\mu_R(M)$ the cardinality of a minimal generating set of $M$ as an $R$--module.

\begin{df}[The rank metric] \cite[Definition 3.3]{rankmetriccodesoverpir}
Let $R$ be a finite chain ring, and let $A\in M_{m,n}(R).$ The \textit{rank} of $A$ is defined as $\rank(A):=\mu_R(\CS(A))$. Let $\cC\subset M_{m,n}(R)$ be a code. For $A,B\in\cC$, the rank distance between $A$ and $B$ is defined $d(A,B)=\rank(A-B)$.
\label{rango1}
\end{df}

\begin{ex}
Consider the Galois ring $R:=\GR(4,3)=\mathbb{Z}[x]/(4,x^3+x+1)$ and consider the code $\mathcal{C}\subseteq\mathrm{M}_{2,2}(R)$ generated by the matrices $\left(\begin{array}{cc}1 & 0\\0 & 0\end{array}\right)$ and $\left(\begin{array}{cc}0 & 0\\0 & 1\end{array}\right)$, namely
$$
\mathcal{C}=\left\{ \left(\begin{array}{cc}a & 0\\0 & b\end{array}\right), a,b\in R\right\}.
$$
This is a free code of rank $2$ over $R$. From the very definition, for a matrix of the form $A=\left(\begin{array}{cc}a & 0\\0 & b\end{array}\right)\in\mathcal{C}$, $\rank(A)=2$ if and only if $a,b\neq 0$, whilst if either $a=0$ or $b=0$ (but not both), then $\rank(A)=1$. In particular, for $a=b=2$ the rank is $2$. Notice that, however, the matrix is not linearly independent, as the entries are zero divisors.
\label{example after defining rank metric}
\end{ex}

It is not complicated to check that the rank distance satisfies the usual axioms and hence, the pair $(\cC, d)$ is a metric space, to which we will refer as a rank metric code over the Galois ring $R$.

For matrices with coefficients on a chain ring, we have the following result:

\begin{prop}\cite[Proposition 3.4 and Corollary 3.5]{rankmetriccodesoverpir}.\label{over a chain ring the rank is as in vector spaces}
Let $R$ be a chain ring, and let $A\in M_{m,n} (R).$ Then, there is an $R$--module isomorphism $\CS(A)\cong\RS(A).$ Moreover, we have $\rank(A)=\mu_R(\RS(A)).$
\end{prop}

We recall next some basic facts from the theory of matrices over rings.

\begin{df}\label{equivalence of matrices}
Let $R$ be a ring, let $m\leq n$ be integers, and let $A,B\in M_{m,n}(R).$ The matrices $A$ and $B$ are said to be equivalent if $B=PAQ$ for some $P\in\operatorname{GL}_m (R)$ and $Q\in\operatorname{GL}_n (R);$ in this case, we write $A\approx B.$
\end{df}

As in the case of finite fields, the minimal distance plays a crucial role in the setting of rank metric codes over Galois rings:

\begin{df}Let $\cC$ be a linear code endowed with the rank metric. The minimal distance of $\cC$ is $d(\cC):=\min\{\rank(U)|U\in\cC\setminus\{O\}\}$, where $O$ denotes the zero matrix.
\end{df}

Recall the $q$-binomial coefficient:
\[
\binomii{k}{k'}_q:= 
\begin{cases}
    0 & \text{if }k< k' \\
    1 & \text{if } k'=0 \\
    \prod\limits_{i=0}^{k'-1}\dfrac{q^k-q^i}{q^{k'}-q^i} & \text{otherwise.}
\end{cases}
\]
We will use the following result:

\begin{lm} \cite[Lemma 1]{countingsumbodulesfrobeniusrings}\label{counting free submodules in an Artinian local ring}
Let $R$ be a finite local ring with maximal ideal $\mathfrak{m}_R$ and finite residue field of $q$ elements. Then, the number of free $R$--linear submodules of rank $k'$ of a given free $R$--linear code of rank $k$ is given by
\[
\stirling{k}{k'}:=\lvert\mathfrak{m}_R\rvert^{k'(k-k')}\binomii{k}{k'}_q.
\]
\end{lm}
In the case of a Galois ring we have:

\begin{cor}\label{counting free submodules in a Galois ring}
Let $R=\GR (\overline{q},m)$.Then, the number of free $R$--linear submodules of rank $k'$ of a given free $R$--linear code of rank $k$ is given by
\[
\stirling{k}{k'}=\theta_0 (m)^{k'(k-k')}\binomii{k}{k'}_q,
\]
where $\theta_0 (m):=(\overline{q}-\phi (\overline{q}))^m,$ and $\phi$ is Euler's totient function. Here, to avoid confussion we recall that $q=p^m$.
\end{cor}

\begin{proof}
Let $\mathfrak{m}$ be the maximal ideal of $R,$ we only need to check that $\abs{\mathfrak{m}}=\theta_0 (m).$ Indeed, let $\beta\in\GR (\overline{q},m),$ set $R_1:=\GR (\overline{q},1),$ and let $\{1=\gamma_0,\gamma_1,\ldots,\gamma_{m-1}\}$ be a free $R_1$--basis of $R$. In this way, we can write \cite[page 144, (3.1)]{Galoisringsbook}
\[
\beta=c_0+\sum_{i=1}^{m-1}c_i\gamma_i
\]
for some unique $c_i\in R_1$. Hence $\beta\in\mathfrak{m}$ if and only if all the $c_i$'s are divisible by $p.$ This shows $\abs{\mathfrak{m}}=\theta_0 (m),$ as claimed.
\end{proof}

The following notations will be useful later on:
\begin{nt}Set
\[
\{n\}:=\stirling{n}{1}\text{ and }\{n\}!:=\prod_{j=1}^n \{j\}.
\]
Notice that
\[
\stirling{k}{k'}=\frac{\{k\}!}{\{k'\}!\cdot\{k-k'\}!}.
\]
\end{nt}
We will also make use of the following result on the structure of modules and submodules of a chain ring.

\begin{teo}\cite[Theorem 2.5]{HonoldLandjev2000}\label{over a chain ring free modules is a serre category}
Let $H$ be  free module of rank $n$ over a chain ring $R$ and let $M$ be a submodule of $H$. Then 
$$M \text{ is free } \iff H/M \text{ is free } \iff \rank(H/M)=n-k.$$    
\end{teo}

Now, assume that $M$ and $N\leq M$ are free submodules of a module over a Galois ring. Hence, by Theorem \ref{over a chain ring free modules is a serre category} the quotient $M/N$ is free and hence  the short exact sequence $0\to N\to M\to M/N\to 0$ splits, namely, $M\cong N \oplus M/N$, a fact that actually happens in any semisimple category.

The following result is a Steinitz type Lemma for free modules over a chain ring:

\begin{cor}\cite[Corollary 6 (c)]{GrassmannformulaoverArtinianrings}\label{Basis extension Lemma}
Let $R$ be a chain ring, let $M$ be a free $R$--module, and let $N\subseteq M$ be a free $R$--submodule of $M.$ Then, if $\{x_1,\dots , x_n\}$ is a free $R$--module basis of $N$ then there are $y_{n+1}\ldots, y_m\in M$ such that $\{x_1,\dots , x_n,y_{n+1},\dots,y_{m} \}$ is a free $R$--module basis of $M.$
\end{cor}

A consequence of Theorem \ref{over a chain ring free modules is a serre category} is that free modules over a chain ring behave somehow as vector spaces of finite dimension over a field:

\begin{cor}\label{for chain rings rank is inclusion preserving}
Let $R$ be a chain ring, let $M$ be a free $R$--module, and let $U,V$ be free $R$--submodules. Then, the following assertions hold.

\begin{enumerate}[(i)]

\item If $U\subseteq V$ then $\rank(U)\leq\rank(V)$ 

\item If $U\subseteq V$ and $\rank(U)=\rank(V)$ then $U=V.$
\end{enumerate}
\end{cor}

\begin{proof}
Let $U\subseteq V$ be free $R$--submodules of $M,$ we have the following short exact sequence.
\[
0\to U\to V\to V/U\to 0.
\]
In this way, using Theorem \ref{over a chain ring free modules is a serre category} we have that $\rank (V)=\rank (U)+\rank (V/U).$ Since $\rank (V/U)\geq 0,$ part (i) holds. On the other hand, $\rank (U)=\rank (V)$ if and only if $\rank (V/U)=0,$ which is equivalent to say, since $V/U$ is free by Theorem \ref{over a chain ring free modules is a serre category}, that $V/U=0,$ hence $V=U$ and part (ii) holds too.
\end{proof}

The Singleton bound is also available in the setting of codes over finite chain rings:

\begin{teo}\cite[Proposition 3.20]{rankmetriccodesoverpir}\label{Singleton bound for the rank metric over chain rings}
Let $(R,\mathfrak{m})$ be a finite chain ring, and let $\mathcal{C}\subseteq M_{m,n}(R)$ be a rank code with minimal rank distance $d$. Then, we have that
\[
\lvert\mathcal{C}\rvert\leq\lvert R\rvert^{\min\{m(n-d+1),n(m-d+1)\}}.
\]
Equivalently, $\rank(\mathcal{C})\leq\min\{m(n-d+1),n(m-d+1)\}$. where $\rank(\mathcal{C})$ denotes the cardinality of a basis of $\mathcal{C}$ as $R$--module.
\end{teo}
A code $\mathcal{C}$ for which Singleton inequality is indeed an equality is called a \emph{maximum rank-distance code}, abridged MRD.

\begin{ex}
For the code $\mathcal{C}$ of Example \ref{example after defining rank metric}, the minimal distance is $1$ and the number of codewords is $4^6$. However, the right hand side of Singelton's bound is $4^{12}$.

For the subcode $\mathcal{C}_0\subseteq\mathcal{C}$ consisting on diagonal matrices which are even multiples of the identity, the minimal distance is $2$, the number of codewords is $2^3$ and the right hand side of Singleton's bound is $2^{12}$.
\label{example after Singleton inequality}
\end{ex}

\section{MacWilliams duality theory}\label{section: coefficient MacWilliams}

In this section we give a relation between the weight enumerator coefficients of a code and those of its dual. Namely, we provide a generalisation of the main result in \cite{ravagnani} to the setting of rank metric codes over finite chain rings, which we made more explicit in the case of Galois rings.

To start with, for each ring $R$ and integers $1\leq n\leq m$, the set $M_{m,n}(R)$ of matrices can be regarded as an hermitian left $M_m (R)$-module with sesquilinear form $b$ given by 
$$
b(X,Y):=\operatorname{Tr} (XY^t),
$$ 
where $M_m (R)$ is regarded as ring with involution given by matrix transposition. 
\begin{df}[The dual of a submodule] Given a submodule $\mathcal{C}\subset M_{m,n}(R)$ (or more in general, just a subset), its dual is defined as
$$
\mathcal{C}^{\perp}:=\{Y\in M_{m,n}(R):\ b(X,Y)=0\text{ for all }X\in\mathcal{C}\}.
$$
\end{df} 
It is easy to see that $\mathcal{C}^{\perp}$ is also a submodule  of $M_{m,n}(R)$. If, in addition, $\mathcal{C}$ is a free submodule, namely, a code, then by \cite[(3.6.2)]{bookquadraticformsoverrings} we have, whenever $b_{|\mathcal{C}}$ is non-singular on $\mathcal{C},$ that
\[
(M_{m,n}(R),b)=(\mathcal{C},b_{|\mathcal{C}})\perp (\mathcal{C}^{\perp},b_{|\mathcal{C}^{\perp}}),
\]
where $\perp$ denotes the orthogonal sum as defined in \cite[(3.4)]{bookquadraticformsoverrings}; in particular, $\mathcal{C}^{\perp}$is also a code, referred to as the dual of $\mathcal{C}$, and we have that
\begin{equation}\label{nondeg}
\rank (\mathcal{C}^{\perp})=mn-\rank(\mathcal{C}).
\end{equation}
\begin{ex}For the code $\mathcal{C}$ of Example \ref{example after defining rank metric}, the dual is $\mathcal{C}^{\perp}=\left\{\left(\begin{array}{cc}0 & a\\b & 0\end{array}\right): a,b\in R\right\}$.

For the code $\mathcal{C}_0$ of Example \ref{example after Singleton inequality}, the dual is $\mathcal{C}_0^{\perp}=\left\{\left(\begin{array}{cc}a & b\\c & d\end{array}\right): a\equiv b\pmod{2}\in R\right\}$.
\label{ejemplo3}
\end{ex}

\subsection{Shortening and strict shortening}
For any finite Artinian local ring $R$, and any free $R$-module $U\subseteq R^m$, we set (compare with \cite[Definition 2.4]{ByrneCotardoRavaganizetafunctions})
\[
\operatorname{Mat}_U(m\times n, R):=\{X\in M_{m,n}(R):\ \CS(X)\subseteq U\}.
\]
The following result is just \cite[Lemma 26 and 27]{ravagnani} replacing a finite field by our ring $R$. Since the proofs are exactly the same as in \cite{ravagnani}, we omit it.

\begin{lm}\label{first Ravagnani calculations}
Let $U\subseteq R^m$ be a free $R$--submodule. Then, we have that
\[
\rank_R (\operatorname{Mat}_U(m\times n, R))=n\rank_R (U)\text{ and }\operatorname{Mat}_U(m\times n, R)^{\perp}=\operatorname{Mat}_{U^{\perp}}(m\times n, R).
\]
\end{lm}

\begin{df}[The weight enumerator]
Let $\cC$ be a $[m\times n, k, d]$ a code of rank $k$ in $M_{m,n}(R)$ of minimal distance $d$. For $d\leq t\leq n$,  denote  $W_t(C)=\lvert\{X\in C:\ \rank(X)=t\}\rvert$. The weight enumerator of $\cC$ is defined as 
$$W_{\cC}(x,y)=x^n+\sum_{t=d}^n W_t(\cC)x^{n-t}y^t.$$
Hereafter, sometimes $W_t(\cC)$ will be denoted as $A_t,$ and $W_t(\cC^{\perp})$ will be denoted by $B_t.$
\end{df}
\begin{ex}
For the code $\mathcal{C}$ of Example \ref{example after defining rank metric}, taking into account that $|R|=4^3$, we have that 
$$
W_{\mathcal{C}}(x,y)=x^2+126xy+3969y^2.
$$
For the code $\mathcal{C}_0$ of Example \ref{example after Singleton inequality}, we have that $W_{\mathcal{C}_0}(x,y)=x^2+7y^2$.
\label{ejemplo4}
\end{ex}

Since in $M_{m,n}(R)$ the trace form is non degenerate we have, according to Eq. \ref{nondeg}, that $\rank(\cC^{\perp})=mn-\rank(\cC)$.

\begin{df}[Shortening and strict shortening]
Let $\cC$ be a code in $M_{m,n}(R)$ and $U\subseteq R^n$ a free submodule.
\begin{enumerate}[(i)]

\item We define the \textit{shortening of $\cC$ by $U$} as $\cC_U := \{X\in\cC:\ U \subseteq\ker (X) \}.$ It is easy to check that $\cC_U$ is an $R$--submodule of $\cC.$

\item We define the \textit{strict shortening of $\cC$ by $U$} as $\widehat{\cC_U}:=\{X\in\cC:\ U = \ker(X) \}.$ In contrast with the shortening, in general $\widehat{\cC_U}$ is not necessarily an $R$--submodule
of $\cC.$
\end{enumerate}
\end{df}

Next technical result will describe in a more explicit way the shortening of a code over a chain ring.

\begin{lm}\label{shortening as intersection}
Let $R$ be a chain ring, let $\cC\subseteq M_{m,n}(R)$ be a rank metric code, and let $U\subseteq R^n$ be a free $R$--submodule of rank $n-u.$ Then, setting $M_U:=\{X\in M_{m,n}(R):\ U\subseteq\ker (X)\},$ we have that $M_U$ is a free $R$--module of rank $um$ such that $\mathcal{C}_U=\mathcal{C}\cap M_U.$
\end{lm}

\begin{proof}
First of all, we consider the short exact sequence
$$
0\to U\to R^n\to R^n/U\to 0.
$$
Since $U$ and $R^n$ are free $R$--modules, by Theorem \ref{over a chain ring free modules is a serre category} we can guarantee that $R^n/U$ is also a free $R$--module of rank $u.$ This implies that $\Hom_R (R^n/U,R^m)$ is a free $R$--module of rank $um,$ and since it is clear that $M_U\cong\Hom_R (R^n/U,R^m),$ we can conclude that $M_U$ is a free $R$--module of rank $um.$ Finally, notice that the equality $\mathcal{C}_U=\mathcal{C}\cap M_U$ is just the definition of shortening.
\end{proof}
Unfortunately, the shortening of a code is in general not free, as next example shows.

\begin{ex}\label{shortening is not always free}
Let $R:=\GR (4,r)$ for an integer $r\geq 1,$ let $m:=2$ and $n:=3.$ We consider as code $\cC\subseteq M_{2,3}(R)$ the free $R$-submodule of rank $1$ generated by
\[
G:=\begin{pmatrix}
2& 1& 0\\
0& 0& 0
\end{pmatrix}.
\]
Now, take $U\subseteq R^3$ as the free rank $1$ submodule generated by $(1,0,0),$ so $U^{\perp}$ is just the $R$-submodule consisting of all vectors in $R^3$ with a zero in the first position. The reader will easily notice that both $U,$ $U^{\perp}$ and $\cC$ are free. On the other hand, we observe that, for any $Y\in\cC$ we have that
\[
\RS (Y)\cap U^{\perp}=\begin{cases}
\langle (0,2,0)\rangle,\text{ if }Y\neq 0,\\
0,\text{ otherwise.}
\end{cases}
\]
In this way, we obtain that $\cC_U$ is the $R$--module generated by
\[
\begin{pmatrix}
0& 2& 0\\
0& 0& 0
\end{pmatrix},
\]
which is not a free $R$--module.
\end{ex}

\subsection{Cardinality of a module and its dual}

In this section, we prove that for any finite chain ring $R$ and for any finitely generated $R$--module $M$, setting $M^*:=\Hom_R (M,R)$, we have $|M|=|M^*|$, a fact we will use to relate the MRD character of a code and its dual.

\begin{teo}\label{the inequality of cardinals holds for chain rings}
Let $R$ be a chain ring, and let $M$ be a finitely generated $R$--module. Then, we have that $M^*$ is non--canonically isomorphic to $M$ and, in particular, $\lvert M^*\rvert=\lvert M\rvert.$
\end{teo}

\begin{proof}
First of all, we denote by $N$ the Jacobson radical of $R.$ Thanks to \cite[Theorem 2.1]{HonoldLandjev2000}, we know that $N=R\Theta$ for some $\Theta\in N\setminus N^2,$ and that $R/N\cong\mathbb{F}_q,$ where $q=p^r,$ $p$ is prime and $r\geq 1$ is an integer. Finally, we denote by $m$ the index of nilpotency of $N,$ so $\lvert R\rvert=q^m.$

In this setting, \cite[Theorem 2.2]{HonoldLandjev2000} tells us that $M$ is isomorphic to a direct sum of cyclic $R$--modules; more precisely, we have
\[
M\cong\frac{R}{N^{\lambda_1}}\oplus\ldots\oplus\frac{R}{N^{\lambda_r}},
\]
where $(\lambda_1,\ldots,\lambda_r)$ is a partition of $\log_q (\lvert M\rvert),$ meaning that $\lambda_i\geq\lambda_{i+1},$ that $\lambda_r\neq 0$ and that
\[
\log_q (\lvert M\rvert)=\lambda_1+\ldots+\lambda_r.
\]
On the other hand, we plan to use the well known isomorphism
\[
\Hom_R (R/I,R/J)\cong\frac{(J:_R I)}{J}
\]
in the following way. Notice that the next equality is also stated in \cite[Proposition 1.2]{Bullingtonspir}.
\[
\Hom_R (R/N^{t},R)\cong\frac{((0):_R N^t)}{(0)}=N^{m-t}.
\]
Now, we have to distinguish three different cases.
\begin{enumerate}[(i)]

\item If $M$ is free, then $M\cong R^{\mu_1},$ where $\mu_1=\dim_{R/N}(M/\Theta M).$ In this case, we also have that $M^*\cong R^{\mu_1},$ and therefore $\lvert M\rvert=\lvert M^*\rvert$ holds.

\item If $M$ is not free, but admits a free direct summand. In this case, we can write
\[
M\cong F\oplus M',
\]
where $F$ is a free, finitely generated $R$--module, and $M'$ does not admit any free direct summand. In this case, thanks again to \cite[Theorem 2.2]{HonoldLandjev2000} we can write
\[
M'\cong \frac{R}{N^{\lambda'_1}}\oplus\ldots\oplus\frac{R}{N^{\lambda'_t}},
\]
where $\lambda'_1<m.$ So, thanks to
\[
\Hom_R (R/N^{t},R)\cong\frac{((0):_R N^t)}{(0)}=N^{m-t},
\]
we have that $M^*\cong F^*\oplus N^{m-\lambda'_1}\oplus\ldots\oplus N^{m-\lambda'_t},$ which in turn is equivalent to
\[
M^*\cong F^*\oplus \frac{R}{N^{\lambda'_1}}\oplus\ldots\oplus\frac{R}{N^{\lambda'_t}}.
\]
So, in this case it is also clear that $\lvert M^*\rvert=\lvert M\rvert.$

\item If $M$ is not free and does not admit any free direct summand, then by the same argument done in case (ii) we have that
\[
M^*\cong\frac{R}{N^{\lambda_1}}\oplus\ldots\oplus\frac{R}{N^{\lambda_r}},
\]
and therefore, in particular, we have $\lvert M^*\rvert=\lvert M\rvert.$

\end{enumerate}
The proof is therefore completed.
\end{proof}

\begin{rk}\label{rowscolumns}
We are generalising Ravagnani's approach to obtain an expanded coefficient-wise version of MacWilliams identity, mimicking the counting arguments provided in \cite{ravagnani} but filling in the points where $R$ being an ring, rather than a field, imposes an extra difficulty. However, we are using several counting formulas as in \cite{countingsumbodulesfrobeniusrings} which we must fit in our overall counting procedure. The problem is that the submodules counted in \cite{ravagnani} correspond to the adjoints of the submodules counted in \cite{countingsumbodulesfrobeniusrings}. As we will see in the next subsection, this is not a serious problem, the reason being that for any ring $R$ and integers $1\leq m\leq n$, we have an isomorphism of free $R$--modules $(-)^t:M_{m,n}(R)\to M_{n,m}(R)$ given by matrix transposition.
\end{rk}

Under this isomorphism, for any free $R$--submodule $U\subseteq R^n,$ the $R$--submodule
\[
R_U:=\{X\in M_{m,n}(R):\ \RS(X)\subseteq U\}
\]
corresponds to $\operatorname{Mat}_U (n\times m,R).$ In particular, they have the same rank and the same cardinality.

In the same way, given $\cC\subseteq M_{m,n}(R)$ a code, we have that the adjoint code of $\cC$
\[
\cC^t:=\{Y\in M_{n,m}(R):\ Y^t\in\cC\}
\]
corresponds to $\cC$ by transposition; in particular, once again $\cC$ and $\cC^t$ have the same rank and cardinality. On the other hand, it is worth to observe the following:
\begin{lm}Setting $M_U:=\{X\in M_{m,n}(R):\ U\subseteq\ker (X)\}$, we have that
$$
M_U^t=\Mat_{U^{\perp}}(n\times m,R)\mbox{ and }\cC_U^t=\cC^t\cap\Mat_{U^{\perp}}(n\times m,R).
$$\label{adjlem}
\end{lm}
\begin{proof}This is immediate since 
\[
M_U^t=\{Y\in M_{n,m}(R):\ Y^t\in M_U\}=\{Y:\ \ker (Y^t)^{\perp}\subseteq U^{\perp}\}=\{Y:\ \CS(Y)\subseteq U^{\perp}\}.\qedhere
\]
\end{proof}
\begin{rk}Notice that $\cC^t\cap\Mat_{U^{\perp}}(n\times m,R)$ is denoted $\cC^t (U^{\perp})$ in \cite[Definition 2.4]{ByrneCotardoRavaganizetafunctions}.
\end{rk}
The next technical result will also be useful in what follows.

\begin{lm}\label{row space can be viewed as a shortening}
Let $R$ be a chain ring, let $\cC\subseteq M_{m,n}(R)$ be a rank metric code, and let $U\subseteq R^n$ be a free $R$--submodule of rank $n-u.$ Then, setting
\[
M_U:=\{X\in M_{m,n}(R):\ U\subseteq\ker (X)\},\quad R_U:=\{X\in M_{m,n}(R):\ \RS(X)\subseteq U\},
\]
we have that $R_U=M_{U^{\perp}}$ and $\mathcal{C}_{U^{\perp}}^{\perp}=\mathcal{C}^{\perp}\cap M_{U^{\perp}}.$
\end{lm}

\begin{proof}
First of all, notice that
\[
R_U=\{X\in M_{m,n}(R):\ \RS(X)\subseteq U\}=\{X\in M_{m,n}(R):\ U^{\perp}\subseteq\RS(X)^{\perp}\}.
\]
Moreover, since $\RS(X)^{\perp}=\ker (X)$ we have that $R_U=\{X\in M_{m,n}(R):\ U^{\perp}\subseteq\ker (X)\}=M_{U^{\perp}},$ as claimed. On the other hand, the equality $\mathcal{C}_{U^{\perp}}^{\perp}=\mathcal{C}^{\perp}\cap M_{U^{\perp}}$ follows immediately from Lemma \ref{shortening as intersection}.
\end{proof}

The next result may be regarded as a generalization of \cite[Lemma 28]{ravagnani} in the setting of chain rings.

\begin{lm}\label{Lemma 28 ravagnani: correct proof}
Let $R$ be a chain ring, and let $U\subseteq R^n$ be a free $R$--submodule of rank $n-u.$ Then, we have 
\begin{equation*}
|\cC_U|=\dfrac{|\cC||\cC^{\perp}_{U^{\perp}}|}{|R|^{m(n-u)}}
\end{equation*}
\end{lm}
\begin{proof}
Begin by defining $R_U:=\{ X\in M_{m,n}(R) | \RS(X)\subseteq U \}$. Next, define the bilinear map $b:\cC\times R_U\to R$ by the assignment $(X,Y)\longmapsto\tr(XY^t).$ We claim that the left nullspace is the left kernel of $b.$ Indeed, we have
\[
\lker(b)=\{X\in\cC| \forall Y\in R_U, \tr(XY^t)=0\}=\{X\in\cC| \forall y\in U, y\in\ker (X)\}.
\]
On the other hand, we can compute the right kernel as follows:
\[
\rker(b)=\{Y\in R_U| \forall X\in\cC, \tr(XY^t)=0\}=R_U\cap\cC^{\perp}= \{Y\in \cC^{\perp}| U^{\perp}\subseteq\ker (Y)\}= \cC^{\perp}_{U^{\perp}},
\]
where notice that the equality $\{Y\in  \cC^{\perp}| U^{\perp}\subseteq\ker (Y)\}= \cC^{\perp}_{U^{\perp}}$ follows immediately from Lemma \ref{row space can be viewed as a shortening}. Therefore $\lker(b)=\cC_U$ and $\rker(b)=\cC^{\perp}_{U^{\perp}}$ and thus the bilinear map induced by $b$
\[
\frac{\cC}{\cC_U} \times \frac{R_U}{\cC^{\perp}_{U^{\perp}}} \rightarrow R
\]
is non-degenerate. Therefore, by adjunction we obtain inclusions $\cC/\cC_U \hookrightarrow  (R_U/\cC^{\perp}_{U^{\perp}})^*$ and $R_U/\cC^{\perp}_{U^{\perp}}\hookrightarrow (\cC/\cC_U)^*$.

In particular, we have $\lvert \cC/\cC_U\rvert\leq\lvert(R_U/\cC^{\perp}_{U^{\perp}})^*\rvert$ and $\lvert(R_U/\cC^{\perp}_{U^{\perp}})\rvert\leq\lvert (\cC/\cC_U)^*\rvert.$ On the other hand, Proposition \ref{the inequality of cardinals holds for chain rings} shows that for any finitely generated $R$--module $M$ we have $\lvert M\rvert=\lvert M^*\rvert$. In this way, combining all these facts we have
\[
\lvert \cC/\cC_U\rvert\leq\lvert(R_U/\cC^{\perp}_{U^{\perp}})^*\rvert=\lvert(R_U/\cC^{\perp}_{U^{\perp}})\rvert\leq\lvert (\cC/\cC_U)^*\rvert=\lvert\cC/\cC_U\rvert.
\]
Summing up, we have $\lvert \cC/\cC_U\rvert=\lvert(R_U/\cC^{\perp}_{U^{\perp}})\rvert,$ and therefore we finally conclude
\[
\frac{\abs{\cC}}{\abs{\cC_U}}=\frac{\abs{R_U}}{\abs{\cC^{\perp}_{U^{\perp}}}},
\]
as needed.
\end{proof}
The reader familiar with \cite[Lemma 28]{ravagnani} could ask why Lemma \ref{Lemma 28 ravagnani: correct proof} recovers and extends op.cit, mainly the submodule denoted $R_U$ therein corresponds, under transposition, with our $\Mat_U(n\times m,R)$ and viceversa. However, both results are equivalent, as we advanced in Remark \ref{rowscolumns} and prove next. 

\begin{lm}\label{true version of Ravagnani Lemma 28}
Let $R$ be a chain ring, and let $U \subseteq R^n$ be a free $R$--submodule of rank $n-u.$ Then, we have that
\[
\abs{\cC\cap\Mat_U (n\times m,R)}=\frac{\abs{\cC}\abs{\cC^{\perp}\cap\Mat_{U^{\perp}}(n\times m,R)}}{\abs{R}^{mu}}.
\]
Equivalently, using the notation established in \cite[Definition 2.4]{ByrneCotardoRavaganizetafunctions} we have
\[
\abs{\cC (U)}=\frac{\abs{\cC}\abs{\cC^{\perp}(U^{\perp})}}{\abs{R}^{mu}}.
\]
\end{lm}

\begin{proof}
First of all, from Remark \ref{rowscolumns} and Lemma \ref{adjlem}, we see that the adjoint of the code $\cC\cap\Mat_U (n\times m,R)$ is precisely $(\cC^t)_{U^{\perp}}.$ Now, using Lemma \ref{Lemma 28 ravagnani: correct proof} we obtain
\[
\abs{\cC\cap\Mat_U (n\times m,R)}=\abs{(\cC^t)_{U^{\perp}}}=\frac{\abs{\cC^t}\abs{((\cC^t)_{U})^{\perp}}}{\abs{R}^{mu}}.
\]
Moreover, we have that $\abs{\cC}=\abs{\cC^t}$ and, on the other hand, from Lemma \ref{first Ravagnani calculations} we obtain that the adjoint of $((\cC^t)_{U})^{\perp}$ is $\cC^{\perp}\cap\Mat_{U^{\perp}}(n\times m,R).$ Taking into account these facts we finally conclude that
\[
\abs{\cC\cap\Mat_U (n\times m,R)}=\frac{\abs{\cC}\abs{\cC^{\perp}\cap\Mat_{U^{\perp}}(n\times m,R)}}{\abs{R}^{mu}},
\]
as claimed.
\end{proof}

\subsection{Binomial moments of the rank distributions of a code and its dual}
In this subsection we define the binomial moments of the rank distribution and put the previous results together to relate the coefficients of the weight enumerator of a code over a Galois ring with those of its dual, generalising \cite{ravagnani} to this setting. Moreover, as a byproduct we will show that the dual of an MRD code over a Galois ring is also MRD. 

\begin{df}
Let $u\geq 0$ be a non--negative integer, then the binomial moment of $\cC$ is
\[
B_u(\cC)=\sum_{\substack{\rank(U)=n-u\\ U \text{ free in } R^n}}(\lvert\cC_U\rvert-1).
\]
\end{df}

The next technical result is a quite direct consequence of Lemma \ref{counting free submodules in an Artinian local ring}.

\begin{lm}\label{Lemma 29 Rav}
Let $R$ be a finite local ring with maximal ideal $\mathfrak{m}$ and residue field of $q$ elements, let $0\leq t,s \leq k$ and let $X\subseteq R^k$ free with $\rank(X)=t$. The number of free submodules $U\subseteq R^k$ such that $X\subseteq U$ and $\rank(U)=s$ is $$\stirling{k-t}{s-t}.$$
\end{lm}

\begin{proof}
Let $U\subseteq R^k$ be a free $R$--submodule such that $X\subseteq U$ and $\rank (U)=s.$ Consider the following short exact sequence.
\[
0\to X\to U\to U/X\to 0.
\]
By Theorem \ref{over a chain ring free modules is a serre category}, $U/X$ is a free $R$--submodule of $R^k/X$ of rank $s-t.$ This shows, in particular, that there is a bijection between free submodules $U\subseteq R^k$ such that $X\subseteq U$ and $\rank(U)=s$ with free submodules $U/X$ of $R^k/X$ with $\rank (R^k/U)=s-t.$ And the cardinality of this set is, by Lemma \ref{counting free submodules in an Artinian local ring}, exactly
\[
\stirling{k-t}{s-t},
\]
as claimed.
\end{proof}

\begin{lm}\label{this is Lemma 30 in Ravagnani}
Let $R$ be a  chain ring, let $\cC\subseteq M_{m,n}(R)$ be a free $R$--linear rank metric code with $(A_i)$ rank distribution, let $0\leq s\leq m$, and for each free $R$--submodule $U\subseteq R^m,$ set $\operatorname{Mat}_U(m\times n, R):=\{M\in M_{m,n}(R)| \CS(M)\subseteq U\}.$ Then, we have that
\begin{equation}\label{equality in Ravagnani Lemma 30}
\sum_{\substack{U\subseteq R^m\\ \rank(U)=s}} \lvert\cC\cap\operatorname{Mat}_U(m\times n, R)\rvert= \sum_{i=0}^m A_i \stirling{m-i}{m-s}.
\end{equation}
\end{lm}

\begin{proof}
Set $A(\cC,s):=\{(U,M)|U\subseteq R^m, \rank(U)=s,M\in \cC,\  \CS(M)\subseteq U\}.$ Now, we count the size of $A(\cC,s)$ in two ways. On the one hand, by Lemma \ref{Lemma 29 Rav}, 
\begin{align*}
|A(\cC,s)|&=\sum_{M\in\cC} |\{U\subseteq R^m| \rank(U)=s, \CS(M)\subseteq U\}| \\
&=\sum_{i=0}^m\sum_{\substack{M\in\cC\\ \rank(M)=i}} |\{U\subseteq R^m| \rank(U)=s, \CS(M)\subseteq U\}|\\
&=\sum_{i=0}^m\sum_{\substack{M\in\cC\\ \rank(M)=i}} \stirling{m-i}{s-i} =\sum_{i=0}^m A_i \stirling{m-i}{s-i} =\sum_{i=0}^m A_i \stirling{m-i}{m-s}.
\end{align*}
On the other hand,
\[
|A(\cC,s)|=\sum_{\substack{U\subseteq R^m\\\rank(U)=s}}\{M\in\cC| \CS(M)\subseteq U\}=\sum_{\substack{U\subseteq R^m\\\rank(U)=s}}|\cC\cap\operatorname{Mat}_U(m\times n, R)|.
\]
Therefore, we have shown that \eqref{equality in Ravagnani Lemma 30} holds, as claimed.
\end{proof}

The next result may be regarded as a generalization of \cite[Proposition 4]{GadouleauYanwithproofs} and of \cite[Theorem 31]{ravagnani}; it relates the binomial moments of the rank distributions of $\mathcal{C}$ and $\mathcal{C}^{\perp}$.

\begin{teo}[Binomial moments for the rank distribution]\label{our MacWilliams identity}
Let $R$ be a  chain ring, and let $\cC\in M_{m,n} (R)$ be a rank metric code. Moreover, let $(A_i)$ and $(B_i)$ be the weight distributions of $\cC$ and $\cC^\perp$ respectively. Then, for any $0\leq\nu\leq n,$ we have
\[
\sum_{i=0}^{n-\nu} A_i \stirling{n-i}{\nu} = \dfrac{|\cC|}{|R|^{m \nu} }\sum_{j=0}^{\nu} B_j \stirling{n-j}{\nu-j}.
\]
\end{teo}

\begin{proof}
First, applying Lemma \ref{this is Lemma 30 in Ravagnani} with $s=n-\nu,$ we have 
\[
\sum_{\substack{U\subseteq R^n\\ \rank(U)=n-\nu}}\lvert\cC\cap \Mat_U(m\times n, R)\rvert= \sum_{i=0}^n A_i \stirling{n-i}{\nu}.
\]
Now, the orthogonality assignment $U\rightarrow  U^\perp$ establishes a bijection between $\rank$ $\nu$ and $\rank$ $n-\nu$ submodules in $R^n$. In this way, we have
\[
\sum_{\substack{U\subseteq R^n\\ \rank(U)=n-\nu}}|\cC^{\perp}\cap \Mat_U^{\perp}(m\times n, R)| =\sum_{\substack{U\subseteq R^n\\ \rank(U)=\nu}}|\cC^{\perp}\cap \Mat_U(m\times n, R)|= \sum_{j=0}^n B_j \stirling{n-j}{n-\nu},
\]
where the last equality follows from Lemma \ref{this is Lemma 30 in Ravagnani} with respect to $\cC^{\perp}$ and $s=\nu.$

Now, using Lemma \ref{Lemma 28 ravagnani: correct proof} with $s=n-\nu,$ we have
$$\sum_{i=0}^n A_i \stirling{n-i}{\nu}=\dfrac{|\cC|}{|R|^{m\nu}}\sum_{j=0}^n B_j\stirling{n-j}{\nu-j},$$
and since for $i>n-\nu, j>\nu, \stirling{n-i}{\nu}=\stirling{n-j}{\nu-j}=0 $ we have our desired result.
\end{proof}
One elementary consequence of Theorem \ref{our MacWilliams identity} is the following statement, which recovers and extends \cite[Corollary 33]{ravagnani}.

\begin{cor}\label{Corollary 33 in Ravagnani}
Let $\cC$ be an $[m\times n,k,d]$ code over a  chain ring $R,$ and let $(A_i)_i,$ $(B_j)_j$ be respectively the rank distributions of $\cC$ and $\cC^{\perp}.$ Given $0\leq\nu\leq n,$ set
\[
a(\nu,n):=\dfrac{\lvert R\rvert^{m \nu}}{\lvert\cC\rvert}\sum_{i=0}^{n-\nu}A_i\stirling{n-i}{\nu}.
\]
Then, the $B_j$'s are given by the recursive formula
\[
B_0=1,\quad B_{\nu}=a(\nu,n)-\sum_{j=0}^{\nu-1}B_j \stirling{n-j}{\nu-j}\text{ if }1\leq\nu\leq n,\quad B_{\nu}=0,\text{ if }\nu>n.
\]
\end{cor}
Finally, we prove, as anounced, that the MRD character is preserved by duality:

\begin{teo}\label{MRD codes are preserved by orthogonality}
If $\cC$ is an MRD code, then so is $\cC^{\perp}.$
\end{teo}

\begin{proof}
Let $\cC\subseteq M_{m,n}(R)$ be an MRD code with $0<\rank(\cC)<mn, m\leq n$, $d=d_{\cC}=$minimum rank, $|\cC|=|R|^{m(n-d+1)}$.
Let $(A_i)$ and $(B_i)$ be the rank distributions of $\cC$ and $\cC^{\perp}$ respectively. $A_0=B_0=1$ and $A_i=0$ for $1\leq i\leq d-1$. By Theorem \ref{our MacWilliams identity} with $\nu=n-d+1$ we have
\[
\stirling{n}{n-d+1}=\stirling{n}{n-d+1}+\sum_{j=1}^{n-d+1} B_j \stirling{n-j}{n-d+1-j}
\]
Since $d\geq 1$ for $1\leq j\leq n-d+1$ then $n-1\geq n-d+1-j\geq 0$ and so $\stirling{n-j}{n-d+1-j}>0$ and therefore it must be that $B_j=0$ for $1\leq j\leq n-d+1$ which implies, denoting by $d^{\perp}$ the distance of $\cC^{\perp},$ that $d^{\perp}\geq n-d+2$. By Theorem \ref{Singleton bound for the rank metric over chain rings}, $\cC^{\perp}$ is MRD.
\end{proof}

\section{Conclusions}
In Section \ref{section: coefficient MacWilliams} we have proved the family of MacWilliams identities relating the $q$-binomial moments of every rank-metric code over a finite chain ring and its dual, extending the corresponding result by Ravagnani for rank-metric codes over finite fields. Our proof, like Ravagnani's exploits the non-degeneracy of the trace form, what we had to prove making use of Theorem \ref{the inequality of cardinals holds for chain rings}, a highly non-trivial result of commutative algebra. As a side product, we have proved (again for rank-metric codes over finite chain rings) that a code is MRD if and only if its dual is so. Unfortunately, it is far from clear how to generalise the MacWilliams identity in form of a MacWilliams transform of the original code, in the spirit of Gadouleau-Yan. We leave this open topic for future (ongoing) work, as well as to the use of both formulas to obtain a functional equation for the corresponding zeta function, also an ongoing work.

\end{document}